\documentclass[conference]{IEEEtran}
\usepackage[utf8]{inputenc}

\IEEEoverridecommandlockouts
\usepackage{cite}
\usepackage{amsmath,amssymb,amsfonts}
\usepackage{algorithm}
\usepackage[noend]{algpseudocode}
\usepackage{graphicx}
\usepackage{support-caption}
\usepackage{subcaption}
\usepackage{textcomp}
\usepackage{xcolor}
\usepackage{amsthm}
\usepackage{tikz}
\usepackage{pgfplots}
\usepackage{pgfplotstable}
\usepackage{soul}
\usepackage[font=footnotesize]{caption}
\pgfplotsset{compat=1.14}
\usepackage{relsize}
\usetikzlibrary{patterns}

\newtheorem{lem}{Lemma}

\newtheorem{cor}{Corollary}

\newtheorem{exmp}{Example}
\theoremstyle{definition}

\newcommand{\CA}[0]{{\mathcal{A}}}
\newcommand{\CB}[0]{{\mathcal{B}}}

\newcommand{\CD}[0]{{\mathcal{D}}}

\newcommand{\CF}[0]{{\mathcal{F}}}

\newcommand{\CH}[0]{{\mathcal{H}}}

\newcommand{\CR}[0]{{\mathcal{R}}}

\newcommand{\CT}[0]{{\mathcal{T}}}

\newcommand{\CV}[0]{{\mathcal{V}}}

\newcommand{\CZ}[0]{{\mathcal{Z}}}

\newcommand{\Bh}[0]{{\mathbf{h}}}

\newcommand{\Bv}[0]{{\mathbf{v}}}

\newcommand{\Bx}[0]{{\mathbf{x}}}

\newcommand{\BA}[0]{{\mathbf{A}}}

\newcommand{\BC}[0]{{\mathbf{C}}}

\newcommand{\BR}[0]{{\mathbf{R}}}

\newcommand{\BV}[0]{{\mathbf{V}}}

\newcommand{\SfX}[0]{{\mathsf{X}}}

\newcommand{\MyRect}[2]{(#2-1+0.1,6-#1+0.1) rectangle (#2-0.1,6-#1+1-0.1)}
\newcommand{\FillBlack}[2]{\filldraw[gray!50]\MyRect{#1}{#2}}
\newcommand{\FillGray}[2]{\filldraw[black!70]\MyRect{#1}{#2}}
\newcommand{\subparagraph}{}
\usepackage{titlesec}
\titlespacing\section{3pt}{6pt plus 4pt minus 2pt}{6pt plus 2pt minus 2pt}
\titlespacing\subsection{3pt}{4pt plus 4pt minus 2pt}{4pt plus 2pt minus 2pt}
\titlespacing\subsubsection{3pt}{3pt plus 4pt minus 2pt}{0pt plus 2pt minus 3pt}

\begin{document}

\title{
A Multi-Antenna Coded Caching Scheme\\with Linear Subpacketization
\thanks{This work was supported by the Academy of Finland under grants no. 319059 (Coded Collaborative Caching for Wireless Energy Efficiency) and 318927 (6Genesis Flagship).}}
\date{August 2019}

\author{
\IEEEauthorblockN{
MohammadJavad Salehi\IEEEauthorrefmark{1},
Antti T\"olli\IEEEauthorrefmark{1},
Seyed Pooya Shariatpanahi\IEEEauthorrefmark{2}
}\\ \vspace{-4mm}
\IEEEauthorblockA{\IEEEauthorrefmark{1}Center for Wireless Communications, University of Oulu, Oulu, Finland.}\\ \vspace{-4mm}
\IEEEauthorblockA{\IEEEauthorrefmark{2}School of Electrical and Computer Engineering, University of Tehran, Tehran, Iran.}\\ \vspace{-4mm}
\IEEEauthorblockA{\{fist\_name.last\_name\}@oulu.fi; p.shariatpanahi@ut.ac.ir}
}

\maketitle

\begin{abstract}
Exponentially growing subpacketization is known to be a major issue for practical implementation of coded~caching, specially in networks with multi-antenna communication setups. We provide a new coded caching scheme for such networks,~which requires linear subpacketization and is applicable to any set of network parameters, as long as the multi-antenna gain $L$~is larger than or equal to the global caching gain $t$.
Our scheme includes carefully designed cache placement and delivery algorithms; which are based on circular shift of two generator arrays in perpendicular directions. It also achieves the maximum possible degrees of freedom of $t+L$, during any transmission interval.
%
\end{abstract}

\begin{IEEEkeywords}
Coded Caching,
Multi-Antenna Communications,
Linear Subpacketization
\end{IEEEkeywords}

\section{Introduction}

Network data traffic has been subject to continuous growth during the past years. The total global IP (Internet Protocol) data volume is estimated to exceed $4.8$ Zettabytes ($10^{21}$ bytes) by 2022, from which 71 percent is expected to pass through wireless networks \cite{cisco2018cisco}. Different applications contribute to the wireless data traffic and each of them requires specific networking Key Performance Indicators (KPIs) such as data rate, delay and reliability. With the introduction of new application types for 5G and beyond (e.g. autonomous vehicles, immersive viewing and massive machine-type communications), extreme advancements for all KPI requirements are expected
\cite{6Genesis2019KeyIntelligence,katz20186genesis}. This has imposed serious challenges in various network layers and solving them is one of the main recent research trends.

Among various networking KPIs, data rate is still of prominent importance. This is mainly due to video applications, as they are expected to account for $82\%$ of the global IP data traffic by 2022 \cite{cisco2018cisco}. 5G networks promote data rates of Gigabits per second; and further increase in the achievable data rate will still be a key driver in future wireless networks \cite{katz20186genesis}.
However, increasing wireless data rate is quite challenging and requires new resources to be used.
5G networks are introducing new frequency bands for cellular communications; from which mm-Wave bands are of much interest as they not only provide larger bandwidth, but also enable cell sizes to be decreased (resulting in better frequency reuse) and larger spatial gain of multi-antenna communications to be achieved \cite{boccardi2014five,osseiran2014scenarios}. 


There is another important resource, recently proposed as a promising enabler of increased data rate for future networks. This idea, originally proposed in \cite{maddah2014fundamental}, is known as Coded Caching and enables a global caching gain, proportional to the total cache size in the network, to be achieved in addition to the local caching gain at each cache location. This results in a new resource, i.e. storage, to become available for data networks; which is specially inspiring as the storage prices are constantly declining \cite{leventhal2008flash}.
Interestingly, coded caching suits well for a majority of video-based applications, for which there is a prime request time (there are time intervals with high request rate) and uneven popularity distribution (a small set of files are requested repeatedly). Also as will be discussed later, it is shown that coded caching gain is additive with multi-antenna gain; making it even more desirable for future networks.

Despite its benefits, coded caching still requires major issues to be solved, before it can be practically implemented. In this paper we target one such issue, known as the large subpacketization requirement. The problem is that the number of smaller parts each file should be split into, known as subpacketization, grows exponentially with respect to the user count $K$; making coded caching implementation infeasible, even for moderate network sizes \cite{lampiris2018adding}. Considering networks with multi-antenna communication setups, we show that linear subpacketization growth is indeed possible, as long as the multi-antenna gain $L$ is larger than or equal to the coded caching gain $t$. Specifically, we introduce a novel coded caching scheme, which requires linear subpacketization with respect to primary network parameters
$K$, $L$, $t$;
to achieve the largest possible degrees of freedom of $t+L$ during any single transmission\footnote{In this paper we assume $t$ does not scale with $K$. If $t$ scales with $K$, the growth in subpacketization will be quadratic.}.
The sole feasibility condition of $L \ge t$ enables the scheme to be applied to a large class of networks; and is in line with the recent trend of using larger antenna arrays.
This is a concrete solution to the subpacketization issue of coded caching schemes, making coded caching one step closer to practical implementation in next-generation networks.

In this paper, we use $[K]$ to denote  $\{1,2,...,K\}$ and $[i:j]$ to represent $\{i,i+1,...,j\}$. Boldface upper- and lower-case letters denote matrices and vectors, respectively. $\BV[i,j]$ refers to the element at the $i$-th row and $j$-th column of matrix $\BV$. Sets are denoted by calligraphic letters. For two sets $\CA$ and $\CB$, $\CA \backslash \CB$ is the set of elements in $\CA$ which are not in $\CB$; and $|\CA|$ represents the number of elements in $\CA$.

\section{System Model}
\label{section:system_model}
We consider a multiple input, single output (MISO) broadcast setup, in which a single server communicates with $K$ users over a shared wireless link with the capacity of $f$ bits per channel use. The server has $L$ transmitting antennas and each user is equipped with a single antenna. Full channel state information (CSI) is available at the server; and it has access to a library of $N \ge K$ files, denoted by $\CF$. Each file $W \in \CF$ has a size of $f$ bits, and each user is equipped with a cache memory of size $Mf$ bits. For simplicity, we use a normalized data unit and drop $f$ in our subsequent notations. 

The system operation consists of two distinct phases, placement and delivery. During the placement phase, which takes place at the low network traffic time, cache memories of the users are filled by data from the files in $\CF$. This in done in accordance with a cache placement algorithm, which operates without any prior knowledge of file request probabilities in the delivery phase. We use $\CZ(k)$ to denote the cache contents of user $k$, after the placement phase is completed.

At the beginning of the delivery phase, each user $k$ reveals its requested file $W(k) \in \CF$. Let us define the demand set as $\CD = \{W(k) \mid k\in [K] \}$. Based on $\CD$ and in accordance with a delivery algorithm, the server builds $S$ transmission vectors $\Bx(1), \Bx(2), ..., \Bx(S)$, each with dimensions $L \times 1$ ($S$ is a design parameter depending on network parameters). Transmission vectors are then transmitted in a TDMA fashion, using the array of $L$ antennas.
After $\Bx(s)$ is transmitted, user $k$ receives
\begin{equation}
\label{eq:reception_model}
    y_k(s) = \Bh_k^T \Bx(s) + w_k (s) \; ,
\end{equation}
where $\Bh_k \in \mathbb{C}^{L}$ denotes the $L \times 1$ channel vector (from~$L$ transmitting antennas); and $w_k(s) \sim \mathbb{C}\mathcal{N}(0,1)$ is the observed noise at user $k$ during transmission interval $s$. Transmission vectors are built such that each user $k$ can decode its requested file $W(k)$, using $\CZ(k)$ (its locally cached data) together with $y_k(1),y_k(2),...,y_k(S)$ (data received from the channel). Let~us denote the set of users targeted by $\Bx(s)$ as $\CT(s)$, for which we have $\CT(s) \subseteq [K]$ and $|\CT(s)| = t+L$. 
We assume zero-forcing beamformers $\Bv_{\CR}$ are used to build $\Bx(s)$, where $\CR \subseteq \CT(s)$ and $|\CR| = t+1$;
and $\Bv_{\CR}$ is built such that $\|\Bv_{\CR}\| = 1$ and
\begin{equation}
\label{eq:zeroforce_definition}
    \begin{aligned}
    \Bh_k^T \Bv_{\CR} &\neq 0 \qquad k \in \CR \; , \\
    \Bh_k^T \Bv_{\CR} &= 0 \qquad \CT(S) \backslash \CR \; .
    \end{aligned}
\end{equation}
We also assume that during downlink training, the server sends orthogonal demodulation pilots precoded by $\Bv_{\CR}$; so that each user $k$ is able to estimate the equivalent channels $\Bh_k^T \Bv_{\CR}$, $\forall \CR$.

Delivery time $T$ is defined as the time required for all users to successfully decode their requested files. Cache placement and delivery algorithms should be designed such that the worst case delivery time (with respect to $\CD$) is minimized. Let us denote the worst case delivery time by $T^*$. Following the common practice in the literature, we assume each user requests a different file, in order to find $T^*$. For simplicity, we also use the notation $A \equiv W(1)$, $B \equiv W(2)$, etc., in the examples provided in this paper. 

As cache placement is done without any knowledge of file request probabilities, an efficient strategy is to store equal-sized data portions of all files in the cache memory of each user. Thereby, every user has $\frac{M}{N}$ of each file in its cache memory, and should receive the rest $(1 - \frac{M}{N})$ of its requested file from the server. This results in a total data size of $K(1-\frac{M}{N})$ to be transmitted over the channel. Let us define the global cache ratio (coded caching gain) $t$ as the total cache size in the network normalized by the number of files, i.e. $t = \frac{KM}{N}$; and assume $t$ is an integer. Then the sum rate of the communication, denoted by $R^*$, is defined as
\begin{equation}
\label{eq:sym_rate_def}
    R^* = \frac{K(1-\frac{t}{K})}{T^*} \; .
\end{equation}
As the channel capacity is one (normalized) data unit per channel use, the symmetric rate also represents how many users benefit from each transmission. So we use the term Degree of Freedom (DoF) equivalent to $R^*$.
The goal is then to design cache placement and delivery algorithms such that DoF is maximized.




\section{State-of-the-Art}

\subsection{Coded Caching}
Coded caching is originally proposed by Maddah-Ali and Niesen in \cite{maddah2014fundamental}, where it is shown that DoF of $t+1$ is achievable with subpacketization $\binom{K}{t}$.
This scheme is later extended~in various directions; e.g. decentralized, hierarchical and multi-server coded caching \cite{maddah2015decentralized,karamchandani2016hierarchical,shariatpanahi2016multi}. Interestingly, in \cite{shariatpanahi2016multi} it is shown that coded caching and multi-server gains are additive; and~so DoF of $t+L$ is achievable with $L$ transmitting servers.~However, the scheme of \cite{shariatpanahi2016multi} requires larger subpacketization of
\begin{equation}
\label{eq:multiserver_subpack}
    \binom{K}{t} \binom{K-t-1}{L-1} \; .
\end{equation}
Following the same concept, multi-antenna coded caching with zero-forcing beamformers is later introduced in \cite{shariatpanahi2017multi,shariatpanahi2018physical}. Optimized beamformers are then used in \cite{tolli2018multicast}, to improve the performance at finite-SNR regime. In \cite{tolli2017multi}, interesting methods based on two design parameters $\alpha,\beta$~are introduced to reduce the optimized beamformer design complexity. However the subpacketization is further increased to
\begin{equation}
    \frac{(\alpha-1)!}{(\delta-1)!(\beta-1)!(t+\beta)!^{\delta-1}} \binom{K}{t} \binom{K-t-1}{L-1} \; ,
\end{equation}
where $\delta = \frac{t+\alpha}{t+\beta}$. In summary, the original scheme of \cite{maddah2014fundamental} and its extensions for multi-antenna setups require exponentially growing subpacketization (with respect to $K$ and for fixed $\frac{M}{N}$), which makes the implementation infeasible even for moderate values of $K$ \cite{lampiris2018adding}. Consequently, reducing subpacketization without decreasing DoF has been studied in the literature, both for single- and multi-antenna coded caching.

\subsection{Subpacketization in Single-Antenna Coded Caching}
Subpacketization is well-studied for single-antenna setups. In \cite{shanmugam2016finite} it is shown that decentralized schemes need exponential subpacketization to achieve any sub-linear rate, for constant $\frac{M}{N}$ as $K \rightarrow \infty$. In \cite{yan2017placement} Placement Delivery~Array (PDA) is presented as a systematic approach to reduce subpacketization in centralized schemes. It is shown that the original scheme of \cite{maddah2014fundamental} is in fact a PDA-driven scheme; and is optimal among a symmetric class of schemes known as $g$-regular PDA.

Following \cite{yan2017placement}, in \cite{yan2017placementb} it is shown that for a constant rate $R^*$, a PDA resulting in linear subpacketization does not exists. In \cite{shangguan2018centralized} a sub-exponential subpacketization scheme for fixed $R^*$ and $\frac{M}{N}$ is proposed. In \cite{shanmugam2017coded} Ruzsa-Szem{\'e}redi graphs are used to design coded caching schemes with linear subpacketization as $K \rightarrow \infty$, but with non-constant $R^*$.


\subsection{Subpacketization in Multi-Antenna Coded Caching}
Subpacketization is less studied for multi-antenna setups. Most notable work on this topic is \cite{lampiris2018adding}, in which it is~shown that if $\frac{K}{L}$ and $\frac{t}{L}$ are both integers, any single-antenna scheme with subpacketization $g(K,t)$ has a multi-antenna counterpart; with subpacketization $g(\frac{K}{L}, \frac{t}{L})$ and without any DoF loss ($g$~is a general function). For example, the scheme of \cite{maddah2014fundamental} can be applied to multi-antenna setups, with subpacketization $\binom{K/L}{t/L}$.
Unfortunately, the scheme of \cite{lampiris2018adding} suffers DoF loss (and also increased subpacketization), if either $\frac{K}{L}$ or $\frac{t}{L}$ is non-integer. Specifically, DoF is reduced
by a multiplicative factor (gap), that is bounded above by 2 when $L > t$, and by $\frac{3}{2}$ when $L < t$.


In \cite{salehi2019subpacketization} it is shown that subpacketization can be traded-off with the performance; and a new approach is introduced for selecting subpacketization in a more flexible manner. The results are however limited to the specific case of $K = t+L$. In \cite{lampiris2019bridging} joint reduction of CSI and subpacketization requirements is considered;
and it shown that subpacketization of $L_c \binom{K_c}{t}$ is achievable, where $L_c = \frac{L+t}{t+1}$ and $K_c = \frac{K}{L_c}$. However, the proposed scheme requires both $L_c$ and $K_c$ to be integers; making it applicable to a very specific set of network parameters. Moreover, it results in a DoF loss by a factor of $(1-\frac{t}{K})$.
\subsection{Our Contribution}
We provide a coded caching scheme with DoF $t+L$ and subpacketization $K \times (t+L)$, for any network with $L \ge t$. The provided scheme requires linear subpacketization with respect to all network parameters $K$, $L$ and $t$, as long as the multi-antenna gain is larger than or equal to the coded caching gain.

\section{Cache Placement}
\label{subsection:placement}
Cache placement is based on placement matrices introduced in \cite{salehi2019subpacketization}, which are also special cases of PDA \cite{yan2017placement}. A placement matrix $\BV$ is a $P \times K$ binary matrix ($P$ can be any integer for which $\frac{Pt}{K}$ is an integer);
with $\sum_p \BV[p,k] = \frac{Pt}{K}, \forall k \in [K]$ and $\sum_k \BV[p,k] = t, \forall p \in [P]$.
Here we use a special placement matrix $\BV$ with $P=K$, in which the first row has $t$ consecutive one elements (other elements are zero); and for the other rows, each row is a circular shift of the previous row by one unit. Given $\BV$, we split each file $W$ into $P=K$ smaller parts $W_p$, and each part $W_p$ into $t+L$ smaller parts $W_p^q$. Then for every $p \in [K], k \in [K]$, if $\BV[p,k] = 1$, $W_p^q$ is stored in the cache memory of user $k$, $\forall W \in \CF, q \in [t+L]$.

\begin{exmp}
\label{exmp:placement_matrix}
Assume $K=6$, $t=2$, $L=3$. $\BV$ is built as
\begin{equation}
\BV = 
    \begin{bmatrix}
    1 & 1 & 0 & 0 & 0 & 0 \\
    0 & 1 & 1 & 0 & 0 & 0 \\
    0 & 0 & 1 & 1 & 0 & 0 \\
    0 & 0 & 0 & 1 & 1 & 0 \\
    0 & 0 & 0 & 0 & 1 & 1 \\
    1 & 0 & 0 & 0 & 0& 1
    \end{bmatrix} \; ,
\end{equation}
and subpacketization is $6 \! \times \! 5 \! = \! 30$.
Cache content of user 1 is
\begin{equation*}
    \begin{aligned}
    \CZ(1) = \{ W_1^1, &W_1^2, W_1^3, W_1^4, W_1^5, \\
    &W_6^1, W_6^2, W_6^3, W_6^4, W_6^5 \; \mid \; W \in \CF\} \; ,
    \end{aligned}
\end{equation*}
and cache content of other users can be written accordingly.
\end{exmp}


\section{Delivery}
\label{section:delivery}
\subsection{Graphical Representation}
\label{section:graphical_rep}
Before formal description of the delivery algorithm, we provide a graphical illustration of its operation for the network of Example \ref{exmp:placement_matrix}.
The delivery algorithm operates in $K=6$ rounds and at each round, $K-t=4$ transmission vectors are built; resulting in $S=24$ total transmission intervals. We also assume the demand set is $\CD = \{A,B,C,D,E,F\}$, and ignore the modulation effect for notation clarity.

Graphical illustration for the first and second rounds are provided in Figures
\ref{fig:graph_exmp_r1c1} and \ref{fig:graph_exmp_r2c2}, respectively. In both figures, each matrix column represents a user and each row stands for a file part index. For example, the first column represents user one, and the first row stands for the first part of all files; i.e. $W_1^q$, $\forall W \in \CF, q \in [t+L]$.
A lightly shaded entry in the matrix means the data part is cached at the respective user.~For example,
$W_1^q,W_6^q$ are stored at user 1, $\forall W \in \CF, q \in [t+L]$. Clearly, the cache placement indicated by Figures \ref{fig:graph_exmp_r1c1} and \ref{fig:graph_exmp_r2c2} is equivalent to the placement matrix $\BV$ provided in Example \ref{exmp:placement_matrix}. 

Consequently, a dark shaded entry indicates which index of the requested file is sent to the respective user, during the given transmission interval. For example, in Figure \ref{fig:sub1} the entries $(3,1), (3,2), (1,3), (1,4), (1,5)$ are dark shaded, which means the first transmission vector at round 1, i.e. $\Bx(1)$, includes $A_3^1,B_3^1, C_1^1, D_1^1, E_1^1$; and $\CT(1)=[1:5]$. With $L=3$ antennas, each part can be nulled out at two users; and we have
\begin{equation}
    \begin{aligned}
        \Bx(1) = &A_3^1 \Bv_{\{1,3,4\}} + B_3^1 \Bv_{\{2,3,4\}} \\
            &+ C_1^1 \Bv_{\{1,2,3\}} + D_1^1 \Bv_{\{1,2,4\}} + E_1^1 \Bv_{\{1,2,5\}} \; . \\
    \end{aligned}
\end{equation}
Note that all superscripts are set to 1, as no data is transmitted prior to $\Bx(1)$. According to \eqref{eq:reception_model} and \eqref{eq:zeroforce_definition}, user 1 receives
\begin{equation}
\label{eq:first_transmission}
    \begin{aligned}
        y_1(1) = &A_3^1 \Bh_1^T \Bv_{\{1,3,4\}} +  \underline{C_1^1 \Bh_1^T \Bv_{\{1,2,3\}}} \\
        &+ \underline{D_1^1 \Bh_1^T \Bv_{\{1,2,4\}}} + \underline{E_1^1 \Bh_1^T\Bv_{\{1,2,5\}}} + w_1(1) \; .
    \end{aligned}
\end{equation}
From Example \ref{exmp:placement_matrix}, we know that $W_1^1 \in \CZ(1)$, $\forall W \in \CF$. Also, based to the system model, user 1 can estimate $\Bh_1^T \Bv_{\CR}$, $\forall \CR$. This means user 1 can reconstruct and remove the underlined terms from its received signal in \eqref{eq:first_transmission}; and finally decode $A_3^1$ interference-free. Similarly, users 2, 3, 4, 5 can decode $B_3^1$, $C_1^1$, $D_1^1$, $E_1^1$ respectively, resulting in DoF of $t+L = 5$ for the first transmission interval.

\begin{figure}[ht]
    \centering
    \begin{subfigure}{.22\textwidth}
        \centering
        \begin{tikzpicture}[scale = 0.5]
            \begin{scope}<+->;
            \draw[step=1cm,very thin,black!60] (0,0) grid (6,6);
            \draw[very thick](0,0)to(0,6);
            \draw[very thick](2,0)to(2,6);
            \draw[very thick](0,5)to(6,5);
            \draw[very thick](0,6)to(6,6);
            \end{scope}
            \begin{scope}
            \FillBlack{1}{1};
            \FillBlack{1}{2};
            \FillBlack{2}{2};
            \FillBlack{2}{3};
            \FillBlack{3}{3};
            \FillBlack{3}{4};
            \FillBlack{4}{4};
            \FillBlack{4}{5};
            \FillBlack{5}{5};
            \FillBlack{5}{6};
            \FillBlack{6}{6};
            \FillBlack{6}{1};
            \FillGray{1}{3};
            \FillGray{1}{4};
            \FillGray{1}{5};
            \FillGray{3}{1};
            \FillGray{3}{2};
            \end{scope}
        \end{tikzpicture}
        \caption{Transmission Interval 1}
        \label{fig:sub1}
    \end{subfigure}
    \begin{subfigure}{.22\textwidth}
        \centering
        \begin{tikzpicture}[scale = 0.5]
            \begin{scope}<+->;
            \draw[step=1cm,very thin,black!60] (0,0) grid (6,6);
            \draw[very thick](0,0)to(0,6);
            \draw[very thick](2,0)to(2,6);
            \draw[very thick](0,5)to(6,5);
            \draw[very thick](0,6)to(6,6);
            \end{scope}
            \begin{scope}
            \FillBlack{1}{1};
            \FillBlack{1}{2};
            \FillBlack{2}{2};
            \FillBlack{2}{3};
            \FillBlack{3}{3};
            \FillBlack{3}{4};
            \FillBlack{4}{4};
            \FillBlack{4}{5};
            \FillBlack{5}{5};
            \FillBlack{5}{6};
            \FillBlack{6}{6};
            \FillBlack{6}{1};
            \FillGray{1}{4};
            \FillGray{1}{5};
            \FillGray{1}{6};
            \FillGray{4}{1};
            \FillGray{4}{2};
            \end{scope}
        \end{tikzpicture}
        \caption{Transmission Interval 2}
        \label{fig:sub2}
    \end{subfigure}
    \begin{subfigure}{.22\textwidth}
        \centering
        \begin{tikzpicture}[scale = 0.5]
            \begin{scope}<+->;
            \draw[step=1cm,very thin,black!60] (0,0) grid (6,6);
            \draw[very thick](0,0)to(0,6);
            \draw[very thick](2,0)to(2,6);
            \draw[very thick](0,5)to(6,5);
            \draw[very thick](0,6)to(6,6);
            \end{scope}
            \begin{scope}
            \FillBlack{1}{1};
            \FillBlack{1}{2};
            \FillBlack{2}{2};
            \FillBlack{2}{3};
            \FillBlack{3}{3};
            \FillBlack{3}{4};
            \FillBlack{4}{4};
            \FillBlack{4}{5};
            \FillBlack{5}{5};
            \FillBlack{5}{6};
            \FillBlack{6}{6};
            \FillBlack{6}{1};
            \FillGray{1}{5};
            \FillGray{1}{6};
            \FillGray{1}{3};
            \FillGray{5}{1};
            \FillGray{5}{2};
            \end{scope}
        \end{tikzpicture}
        \caption{Transmission Interval 3}
        \label{fig:sub3}
    \end{subfigure}
    \begin{subfigure}{.22\textwidth}
        \centering
        \begin{tikzpicture}[scale = 0.5]
            \begin{scope}<+->;
            \draw[step=1cm,very thin,black!60] (0,0) grid (6,6);
            \draw[very thick](0,0)to(0,6);
            \draw[very thick](2,0)to(2,6);
            \draw[very thick](0,5)to(6,5);
            \draw[very thick](0,6)to(6,6);
            \end{scope}
            \begin{scope}
            \FillBlack{1}{1};
            \FillBlack{1}{2};
            \FillBlack{2}{2};
            \FillBlack{2}{3};
            \FillBlack{3}{3};
            \FillBlack{3}{4};
            \FillBlack{4}{4};
            \FillBlack{4}{5};
            \FillBlack{5}{5};
            \FillBlack{5}{6};
            \FillBlack{6}{6};
            \FillBlack{6}{1};
            \FillGray{1}{6};
            \FillGray{1}{3};
            \FillGray{1}{4};
            \FillGray{2}{1};
            \FillGray{6}{2};
            \end{scope}
        \end{tikzpicture}
        \caption{Transmission Interval 4}
        \label{fig:sub4}
    \end{subfigure}
    \caption{Graphical Illustration of the First Round}
    \label{fig:graph_exmp_r1c1}
\end{figure}
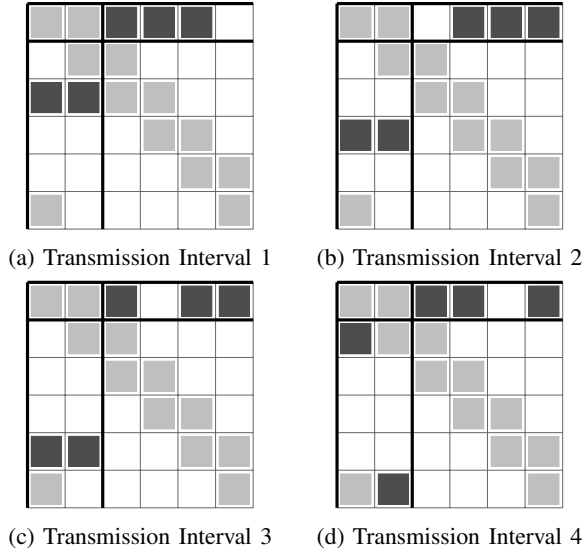

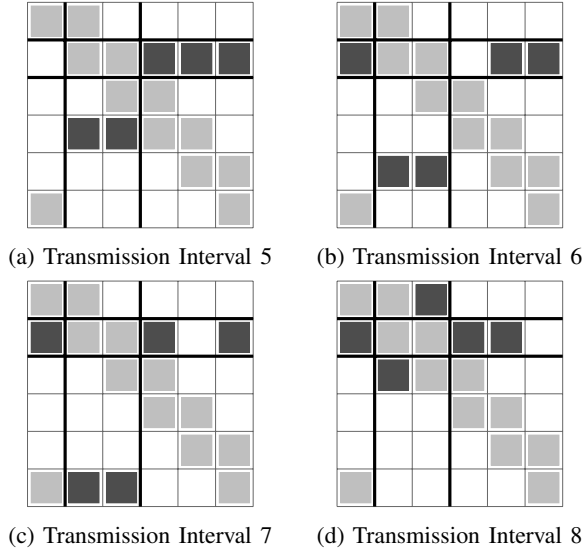
\begin{figure}[ht]
    \centering
    \begin{subfigure}{.22\textwidth}
        \centering
        \begin{tikzpicture}[scale = 0.5]
            \begin{scope}<+->;
            \draw[step=1cm,very thin,black!60] (0,0) grid (6,6);
            \draw[very thick](1,0)to(1,4)to(3,4)to(3,0);
            \draw[very thick](1,6)to(1,5)to(3,5)to(3,6);
            \draw[very thick](0,4)to(1,4)to(1,5)to(0,5);
            \draw[very thick](6,4)to(3,4)to(3,5)to(6,5);
            \end{scope}
            \begin{scope}
            \FillBlack{1}{1};
            \FillBlack{1}{2};
            \FillBlack{2}{2};
            \FillBlack{2}{3};
            \FillBlack{3}{3};
            \FillBlack{3}{4};
            \FillBlack{4}{4};
            \FillBlack{4}{5};
            \FillBlack{5}{5};
            \FillBlack{5}{6};
            \FillBlack{6}{6};
            \FillBlack{6}{1};
            \FillGray{2}{4};
            \FillGray{2}{5};
            \FillGray{2}{6};
            \FillGray{4}{2};
            \FillGray{4}{3};
            \end{scope}
        \end{tikzpicture}
        \caption{Transmission Interval 5}
        \label{fig:sub12}
    \end{subfigure}
    \begin{subfigure}{.22\textwidth}
        \centering
        \begin{tikzpicture}[scale = 0.5]
            \begin{scope}<+->;
            \draw[step=1cm,very thin,black!60] (0,0) grid (6,6);
            \draw[very thick](1,0)to(1,4)to(3,4)to(3,0);
            \draw[very thick](1,6)to(1,5)to(3,5)to(3,6);
            \draw[very thick](0,4)to(1,4)to(1,5)to(0,5);
            \draw[very thick](6,4)to(3,4)to(3,5)to(6,5);
            \end{scope}
            \begin{scope}
            \FillBlack{1}{1};
            \FillBlack{1}{2};
            \FillBlack{2}{2};
            \FillBlack{2}{3};
            \FillBlack{3}{3};
            \FillBlack{3}{4};
            \FillBlack{4}{4};
            \FillBlack{4}{5};
            \FillBlack{5}{5};
            \FillBlack{5}{6};
            \FillBlack{6}{6};
            \FillBlack{6}{1};
            \FillGray{2}{5};
            \FillGray{2}{6};
            \FillGray{2}{1};
            \FillGray{5}{2};
            \FillGray{5}{3};
            \end{scope}
        \end{tikzpicture}
        \caption{Transmission Interval 6}
        \label{fig:sub22}
    \end{subfigure}
    \begin{subfigure}{.22\textwidth}
        \centering
        \begin{tikzpicture}[scale = 0.5]
            \begin{scope}<+->;
            \draw[step=1cm,very thin,black!60] (0,0) grid (6,6);
            \draw[very thick](1,0)to(1,4)to(3,4)to(3,0);
            \draw[very thick](1,6)to(1,5)to(3,5)to(3,6);
            \draw[very thick](0,4)to(1,4)to(1,5)to(0,5);
            \draw[very thick](6,4)to(3,4)to(3,5)to(6,5);
            \end{scope}
            \begin{scope}
            \FillBlack{1}{1};
            \FillBlack{1}{2};
            \FillBlack{2}{2};
            \FillBlack{2}{3};
            \FillBlack{3}{3};
            \FillBlack{3}{4};
            \FillBlack{4}{4};
            \FillBlack{4}{5};
            \FillBlack{5}{5};
            \FillBlack{5}{6};
            \FillBlack{6}{6};
            \FillBlack{6}{1};
            \FillGray{2}{6};
            \FillGray{2}{1};
            \FillGray{2}{4};
            \FillGray{6}{2};
            \FillGray{6}{3};
            \end{scope}
        \end{tikzpicture}
        \caption{Transmission Interval 7}
        \label{fig:sub32}
    \end{subfigure}
    \begin{subfigure}{.22\textwidth}
        \centering
        \begin{tikzpicture}[scale = 0.5]
            \begin{scope}<+->;
            \draw[step=1cm,very thin,black!60] (0,0) grid (6,6);
            \draw[very thick](1,0)to(1,4)to(3,4)to(3,0);
            \draw[very thick](1,6)to(1,5)to(3,5)to(3,6);
            \draw[very thick](0,4)to(1,4)to(1,5)to(0,5);
            \draw[very thick](6,4)to(3,4)to(3,5)to(6,5);
            \end{scope}
            \begin{scope}
            \FillBlack{1}{1};
            \FillBlack{1}{2};
            \FillBlack{2}{2};
            \FillBlack{2}{3};
            \FillBlack{3}{3};
            \FillBlack{3}{4};
            \FillBlack{4}{4};
            \FillBlack{4}{5};
            \FillBlack{5}{5};
            \FillBlack{5}{6};
            \FillBlack{6}{6};
            \FillBlack{6}{1};
            \FillGray{2}{1};
            \FillGray{2}{4};
            \FillGray{2}{5};
            \FillGray{3}{2};
            \FillGray{1}{3};
            \end{scope}
        \end{tikzpicture}
        \caption{Transmission Interval 8}
        \label{fig:sub42}
    \end{subfigure}
    \caption{Graphical Illustration of the Second Round}
    \label{fig:graph_exmp_r2c2}
\end{figure}

The next transmission vectors in round 1, i.e. $\Bx(2)$-$\Bx(4)$, are built by circular shift of $\Bx(1)$ elements over the non-shaded cells of the grid, in two perpendicular directions. Specifically, the first two terms of $\Bx(1)$ are shifted vertically, while the other three terms are shifted horizontally. This procedure is depicted in Figures \ref{fig:sub2}-\ref{fig:sub4}. So $\Bx(2)$ is built as
\begin{equation}
    \begin{aligned}
        \Bx(2) = &A_4^1 \Bv_{\{1,4,5\}} + B_4^1 \Bv_{\{2,4,5\}} \\
            &+ D_1^2 \Bv_{\{1,2,4\}} + E_1^2 \Bv_{\{1,2,5\}} + F_1^1 \Bv_{\{1,2,6\}} \; , \\
    \end{aligned}
\end{equation}
where the superscripts for $D_1$ and $E_1$ are updated to 2, as $D_1^1$ and $E_1^1$ were transmitted by $\Bx(1)$. Similarly, $\Bx(3)$ and $\Bx(4)$ are built as
\begin{equation}
    \begin{aligned}
        \Bx(3) = &A_5^1 \Bv_{\{1,5,6\}} + B_5^1 \Bv_{\{2,5,6\}} \\
            &+ E_1^3 \Bv_{\{1,2,5\}} + F_1^2 \Bv_{\{1,2,6\}} + C_1^2 \Bv_{\{1,2,3\}} \; , \\
        \Bx(4) = &A_2^1 \Bv_{\{1,2,3\}} + B_6^1 \Bv_{\{1,2,6\}} \\
            &+ F_1^3 \Bv_{\{1,2,6\}} + C_1^3 \Bv_{\{1,2,3\}} + D_1^3 \Bv_{\{1,2,4\}} \; .
    \end{aligned}
\end{equation}

The second transmission round includes $\Bx(5)$-$\Bx(8)$, which are built by diagonal shift (simultaneous circular shift of one unit to the right and down)
of $\Bx(1)$-$\Bx(4)$; as shown in Figure~\ref{fig:graph_exmp_r2c2}. Similarly, the third round is built by diagonal shift of $\Bx(5)$-$\Bx(8)$, and this procedure continues until $\Bx(21)$-$\Bx(24)$ are built by diagonal shift of the transmission vectors of the previous round. In total, $6 \times 4 = 24$ transmission intervals are required and each missing part will appear $2+3 = 5$ times; resulting in total subpacketization requirement of $6 \times 5 = 30$.

For a general network setup with parameters $K, t, L$, in a single round we have $K-t$ transmission vectors, and each new round is built by diagonal shift of transmission vectors in the previous round. There exist a total number of $K$ rounds, resulting in $S = K \times (K-t)$ total transmission intervals; and the required subpacketization is $K \times (t+L)$.


\subsection{Delivery Prime Matrices}
\label{section:delivery_prime_matrix}
In order to provide the delivery algorithm, we first introduce and construct Delivery Prime (DP) matrices. Denoted by $\BR_k$ and $\BC_k$, $k \in [K]$, they are a group of matrices with dimensions $(K-t) \times (t+L)$. $\BR_1$ and $\BC_1$ are built using Algorithms \ref{alg:R1_creation} and \ref{alg:C1_creation} (for all algorithms, $K, L, N, t$ are assumed to be global variables known to the procedures);
and for $k > 1$, we use circular increment to build $\BR_k$ and $\BC_k$ from $\BR_{k-1}$ and $\BC_{k-1}$, respectively. For an integer $a$, the intended circular increment operation in domain $K$ is defined as
\begin{equation}
    i_c (a,K) = (a \mod K) + 1 \; ,
\end{equation}
while for a matrix $\BA$ with positive integer elements, $i_c(\BA,K)$ results in a matrix in which each element is the circular increment (in domain $K$) of its respective element in $\BA$. Now for $k \in [2:K]$ we define
\begin{equation}
    \begin{aligned}
    \BR_k = i_c(\BR_{k-1},K) \; ; \;\; \BC_k = i_c(\BC_{k-1},K) \; .
    \end{aligned}
\end{equation}
%
\begin{exmp}
\label{exmp:r1_c1}
For the network of Example \ref{exmp:placement_matrix}, we have
\begin{equation}
\label{eq:r1_exmp_1}
\begin{aligned}
    \BR_1 =
    \begin{bmatrix}
    3 & 3 & 1 & 1 & 1 \\
    4 & 4 & 1 & 1 & 1 \\
    5 & 5 & 1 & 1 & 1 \\
    2 & 6 & 1 & 1 & 1
    \end{bmatrix}, 
    \; \BC_1 =
    \begin{bmatrix}
    1 & 2 & 3 & 4 & 5 \\
    1 & 2 & 4 & 5 & 6 \\
    1 & 2 & 5 & 6 & 3 \\
    1 & 2 & 6 & 3 & 4
    \end{bmatrix}, \\
    \BR_2 =
    \begin{bmatrix}
    4 & 4 & 2 & 2 & 2 \\
    5 & 5 & 2 & 2 & 2 \\
    6 & 6 & 2 & 2 & 2 \\
    3 & 1 & 2 & 2 & 2
    \end{bmatrix},
    \; \BC_2 =
    \begin{bmatrix}
    2 & 3 & 4 & 5 & 6 \\
    2 & 3 & 5 & 6 & 1 \\
    2 & 3 & 6 & 1 & 4 \\
    2 & 3 & 1 & 4 & 5
    \end{bmatrix}.
\end{aligned}
\end{equation}
\end{exmp}

\begin{algorithm}[t]
\caption{$\BR_1$ Generation Procedure}
\label{alg:R1_creation}
\begin{algorithmic}[1]
    \Procedure{Generate $\BR_1$}{}
        \ForAll{$j \in [1:t]$}
            \ForAll{$i \in [1:K-2t+j]$}
                \State $\BR_1 [i,j] \gets t+i$
            \EndFor
            \ForAll{$i \in [K-2t+j+1:K-t]$}
                \State $\BR_1 [i,j] \gets t+i-(K-2t+j-1)-t$
            \EndFor
        \EndFor
        \ForAll{$j \in [t+1:t+L]$, $i \in [1:K-t]$}
            \State $\BR_1 [i,j] \gets 1$
        \EndFor
    \EndProcedure
\end{algorithmic}
\end{algorithm}

\begin{algorithm}[t]
\caption{$\BC_1$ Generation Procedure}
\label{alg:C1_creation}
\begin{algorithmic}[1]
    \Procedure{Generate $\BC_1$}{}
        \ForAll{$j \in [1:t]$, $i \in [1:K-t]$}
            \State $\BC_1 [i,j] \gets j$
        \EndFor
        \ForAll{$j \in [t+1:t+L]$, $i \in [1:K-t]$}
            \If{$j+i-1 \le K$}
                \State $\BC_1 [i,j] \gets j+i-1$
            \Else
                \State $\BC_1 [i,j] \gets j+i-1-(K-t)$
            \EndIf
        \EndFor
    \EndProcedure
\end{algorithmic}
\end{algorithm}

\subsection{The Delivery Algorithm}
Delivery procedure is provided in Algorithm \ref{alg:delivery_main}; with its auxiliary procedures presented in Algorithms \ref{alg:aux_init} and \ref{alg:aux_gen_r}. We have used $\{ \BC_k \} \equiv \{ \BC_1, ..., \BC_K \}$ and $\{ \BR_k \} \equiv \{ \BR_1, ..., \BR_K \}$; and \textsc{Modulate} function returns the modulated version of its input.
The procedure is based on DP matrices. For each $k \in [K]$ we use $\BR_k$ and $\BC_k$ jointly, to create one transmission round. Each round has $K-t$ transmission vectors; resulting in $S = K \times (K-t)$ transmission intervals in total.


As an explanation, for every $k \in [K]$, a transmission~vector $\Bx(s)$ is built for each row in $\BR_k$ and $\BC_k$. More precisely, the entries in each row of $\BC_k$ specify the user indices to be targeted during the transmission interval, i.e. $\CT(s)$; while the entries in the respective row of $\BR_k$ clarify the file parts to be selected for users in $\CT(s)$. Finally, for $n \in [N]$, $k \in [K]$, the superscript $q(n,k)$ is another index which exactly specifies which data portion should be sent during each transmission.


\begin{algorithm}[t]
\caption{Delivery Procedure}
\label{alg:delivery_main}
\begin{algorithmic}[1]
    \Procedure{Delivery}{$\{\BR_k \}, \{\BC_k \}$}
        \State \textsc{Initialize}
        \ForAll{$k \in [K]$}
            \ForAll{$i \in [K-t]$}
                \State $s \gets s+1$
                \State $\Bx(s) \gets 0$
                \ForAll{$j \in [t+L]$}
                    \State $usr \gets \BC_k[i,j]$
                    \State $prt \gets \BR_k[i,j]$
                    \State $ind \gets q \big( W(usr),prt \big)$
                    \State $\SfX \gets $ \textsc{Modulate}$\big( W_{prt}^{ind} (usr) \big)$
                    \State $\CR \gets$ \textsc{Generate $\CR$}($k, i, usr, prt, \{\BC_k \}$)
                    \State $\Bx(s) \gets \Bx(s) + \Bv_{\CR} \SfX$
                    \State $q \big( W(usr),prt \big) \gets q \big( W(usr),prt \big) + 1$
                \EndFor
                \State Transmit $\Bx(s)$
            \EndFor
        \EndFor
    \EndProcedure
\end{algorithmic}
\end{algorithm}

\begin{algorithm}[t]
\caption{Initialization Procedure}
\label{alg:aux_init}
\begin{algorithmic}[1]
    \Procedure{Initialize}{}
        \State $s \gets 0$
        \ForAll{$n \in [N]$, $p \in [K]$}
            \State $q(n,p) \gets 1$
        \EndFor
    \EndProcedure
\end{algorithmic}
\end{algorithm}

\begin{algorithm}[t]
\caption{$\CR$ Generation Function}
\label{alg:aux_gen_r}
\begin{algorithmic}[1]
    \Function{Generate $\CR$}{$k, row, usr, prt, \{\BC_k\}$}
            \State $\CR \gets \{usr \}$
            \ForAll{$r \in [t+L]$}
                \State $node \gets \BC_k[row,r]$
                \If{$\BV[prt,node] = 1$}
                    \State $\CR \gets \CR \cup \{node \}$
                \EndIf
            \EndFor
            \State \Return $\CR$
    \EndFunction
\end{algorithmic}
\end{algorithm}

\section{Validity and Performance}
\label{Section:performance_analysis}
\begin{lem}
\label{lem:dof}
During each transmission interval, exactly $t+L$ users receive part of their requested data, interference free.
\end{lem}
\begin{proof}
According to graphical and algorithmic representations of Section \ref{section:delivery}, the delivery phase consists of $K$ rounds; and each round includes $K-t$ transmission vectors. However, the transmission vectors at each round are diagonal shifts of the ones at the previous round, and cache placement structure is also immune to the diagonal shift operation. As a result, it is enough to show that all transmission vectors at a single round serve $t+L$ users interference free.

Without loss of generality, let us consider the first round. Using $\CT_r^j$ to denote the set of users targeted at transmission~$j$ of round $r$, we have $\CT_1^1 = [1:t+L]$. We distinguish two disjoint subsets of $\CT_1^1$; as $\CV_1^1 = [1:t]$ and ${\CH_1^1 = [t+1:t+L]}$. For any $v \in \CV_1^1$, we transmit $W_t(v)$ (superscript $q$ is ignored for simplicity), which is available at the cache memory of $t$ users in $\CH_1^1$. So, for interference free data delivery, $W_t(v)$ should be nulled out at the other $L-t$ users of $\CH_1^1$; as well as all $t-1$ users of $\CV_1^1 \backslash \{v\}$. On the other hand, for any $h \in \CH_1^1$, $W_1(h)$ is available at the cache memory of all users in $\CV_1^1$; and should be nulled out at all $L-1$ users of $\CH_1^1 \backslash \{h \}$. So every data part in the transmission vector should be nulled out at $L-1$ users in total, which is indeed possible with $L$ transmitting antennas. This means all users in $\CT_1^1$ can get part of their requested data, interference free.


Subsequent transmission vectors in round one are designed by circular shift of users in $\CV_1^1$ and $\CH_1^1$, in vertical and horizontal directions respectively. However, for any $j \in [K-t]$, the file part intended for each user in $\CH_1^j$ is always available in the cache memory of all users in $\CV_1^j$; and should be nulled out at $L-1$ other users of $\CH_1^j$. Moreover, the diagonal structure of the cache placement causes the file part requested by each user in $\CV_1^j$ to be available in the cache memories of $t$ users in $\CH_1^j$; which means it should be also nulled out at $L-1$ total users. So every transmission vector at round one delivers data to $t+L$ users interference free, and the proof is complete.
\end{proof}

\begin{lem}
For each $W \in \CF$, $W_p$ should be split into $t+L$ smaller parts, for the algorithm to work properly.
\end{lem}
\begin{proof}
Using the same notation as the proof of Lemma \ref{lem:dof}, in round $r$, transmission vectors are built by circular shift of $\CV_r^1$ and $\CH_r^1$, in vertical and horizontal directions respectively. As a consequence, during round $r$:
\begin{itemize}
    \item for each user $k \in \bigcup \CH_r^j$, $W_r(k)$ appears $L$ times in the transmission vectors;
    \item for each user $k \in \CV_r^1$ and $p \in [K]$ such that $\BV [p,k] = 0$, $W_p(k)$ appears once in the transmission vectors.
\end{itemize}
Moreover, for any $k,p \in [K]$ with $\BV [p,k] = 0$, there exists exactly one $r$ value for which $W_p(k)$ appears in $\bigcup \CH_r^j$; but there exist $t$ different $r$ values for which $W_p(k)$ appears in $\CV_r^j$ (for some $j$). So $W_p(k)$ appears $L \times 1 + 1 \times t = L+t$ times during all transmissions; which means each $W_p$ should be split in $t+L$ smaller parts for the algorithm to work properly.
\end{proof}

\begin{cor}
The provided coded caching scheme achieves the maximum possible DoF of $t+L$ during each transmission and requires linear subpacketization of $K \times (t+L)$.
\end{cor}

Linear growth in subpacketization enables coded caching to be practically implemented in large networks. For example, if $K=20, L=4, t=2$, the multi-server (MS) scheme of~\cite{shariatpanahi2016multi} requires subpacketization of 129,200; while our algorithm reduces it to 120 (1,000 fold decrease). Keeping $L$ and $t$ fixed and increasing $K$ to 50, the reduction becomes 66,000 fold.

For very small networks however, there exist specific cases where the MS scheme requires smaller subpacketization. For example, for $t=2$, $L=3$, $K \in [5:10]$, subpacketization of both schemes is plotted in Figure \ref{fig:compare_plot}. Clearly, the MS scheme outperforms the new scheme at $K=5$. However, our scheme has smaller subpacketization as~$K$ is increased, and the gap between the two schemes grows exponentially.

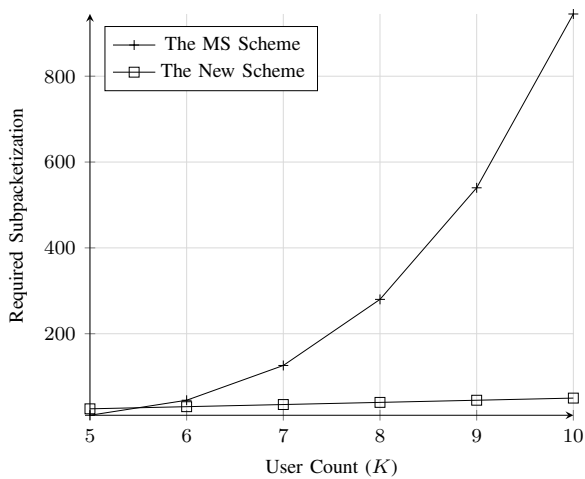
\begin{figure}
    \centering
    \resizebox{0.9\columnwidth}{!}{%

    \begin{tikzpicture}

    \begin{axis}
    [
    axis lines = left,
    xlabel = \smaller {User Count ($K$)},
    ylabel = \smaller {Required Subpacketization},
    ylabel near ticks,
    legend pos = north west,
    ticklabel style={font=\smaller},
    grid=both,
    major grid style={line width=.2pt,draw=gray!30},
    ]
    
    \addplot
    [
    mark = +,
    black
    ]
    table
    [
    y=MultiServer,
    x=K
    ]
    {Figures/Compare_Plot_Data.txt};
    \addlegendentry{\smaller The MS Scheme}
    
    \addplot
    [
    mark = square,
    black
    ]
    table
    [
    y=New,
    x=K
    ]
    {Figures/Compare_Plot_Data.txt};
    \addlegendentry{\smaller The New Scheme}
    
    \end{axis}

    \end{tikzpicture}
    }

    \caption{Subpacketization Comparison - $t=2$, $L=3$}
    \label{fig:compare_plot}
\end{figure}

Compared with the scheme of \cite{lampiris2018adding}, as it has DoF loss if either $\frac{K}{L}$ or $\frac{t}{L}$ is non-integer, 
it is only comparable with our scheme if $t=L$ and at the same time, $\frac{K}{L}$ is an integer. For this special case, it requires subpacketization $\frac{K}{L}$, outperforming our scheme by a factor of $L^3$. However, as mentioned, this is only valid for a very specific class of network parameters and 
otherwise, the two schemes cannot be directly compared.

\section{Conclusion and Future Work}
We proposed a coded caching scheme with linear subpacketization; which is applicable for any set of network parameters as long as the the multi-antenna gain is larger than or equal to the global caching gain. Moreover, it achieves the full additive gain of coded caching and multi-antenna communication, in every transmission interval.

The DoF analysis provided in this paper is applicable only to the high-SNR regime, however. It remains an open problem to analyze the system behavior at finite-SNR, where beamformer design complexity is another issue to be considered alongside the large subpacketization requirement.

Another question is the possibility of constructing the DP matrices (and designing the delivery algorithm respectively), using other values of $P$, i.e. $P \neq K$, for the placement matrix. In fact, increasing $P$ might enable better multicasting opportunity, resulting in increased efficiency index and hence better finite-SNR rate, as outlined in \cite{salehi2019subpacketization}.

Finally, extending the provided scheme to be applicable to a larger selection of network parameters, i.e. the region $t > L$; and improving its subpacketization requirement for $t$ values scaling with $K$, is part of our ongoing research.

\bibliographystyle{IEEEtran}
\bibliography{References}

\end{document}